\documentclass[13pt,a4paper]{article} 
 
 \usepackage{amsmath}
 \usepackage{upquote}
 \usepackage{color}
 \usepackage{xcolor}
 \usepackage{enumerate}
  \usepackage{mathrsfs} 
 \usepackage[document]{ragged2e}
 \usepackage{amssymb}
 \usepackage{mathtools}
 \usepackage{easybmat} 
 \usepackage{amsfonts}
 \usepackage{amsthm}
 \usepackage{lipsum}
 \usepackage{afterpage} 
\usepackage{hyperref}
\usepackage{MnSymbol}
\usepackage{wasysym}
\usepackage{accents}
\usepackage[toc]{appendix} 
\usepackage[thinlines]{easytable} 
\usepackage[font=footnotesize,labelfont=bf]{caption}
\usepackage{cleveref}
\usepackage{pbox}
\usepackage[margin=0.8in]{geometry}

\DeclareCaptionFont{tiny}{\tiny}
\numberwithin{equation}{section}

\newcommand\numberthis{\addtocounter{equation}{1}\tag{\theequation}} 

\usepackage{graphicx}
\usepackage{sectsty}
\usepackage{float}
\usepackage[section]{placeins}
\usepackage{wrapfig}
\usepackage{tocbibind}

\newtheorem{proposition}{Proposition}
\theoremstyle{definition}

\theoremstyle{plain}

\newtheorem{theorem}{Theorem}
\newtheorem{remark}{Remark}
\usepackage{hyperref} 
\usepackage[utf8]{inputenc}
\newtheorem{corollary}{Corollary}

\begin{document}

\title{An Analytical Mechanics Approach to the First Law of Thermodynamics and  Construction of a Variational Hierarchy}

\date{}
\author{ \hspace{-16mm}\\ 
\hspace{0mm} \small HAMID SAID \renewcommand*{\thefootnote}
{\fnsymbol{footnote}}\footnote{\textit{e-mail:}   \textsf{hamids@sci.kuniv.edu.kw}} \hspace{1mm} \\
\hspace{0mm} \footnotesize Department of Mathematics, Kuwait University, PO Box 5969, Safat 13060, Kuwait}
\maketitle

\begin{abstract}
\justify A simple procedure is presented to study the conservation of energy equation with dissipation in continuum mechanics in 1D. This procedure is used to transform this nonlinear evolution-diffusion equation into a hyperbolic PDE; specifically, a second order quasi-linear wave equation. An immediate implication of this procedure is the formation of a least action principle for the balance of energy with dissipation. The corresponding action functional enables us to establish a complete analytic mechanics for thermomechanical systems: a Lagrangian-Hamiltonian theory, integrals of motion, bracket formalism, and Noether's theorem. Furthermore, we apply our procedure iteratively and produce an infinite sequence of interlocked variational principles, a \emph{variational hierarchy}, where at each level or iteration the full implication of the least action principle can be shown again.
\end{abstract} 
\begin{flushleft}

\end{flushleft}

\begin{flushleft}

\end{flushleft}

\section{Introduction}

\begin{flushleft}
\justify 
Hamilton's principle is, undoubtedly, one of the great insights of physics. While historically it was formulated in the context of classical mechanics \cite{gregory2006classical}, it has been remarkably extended to other field theories such as fluid mechanics, electromagnetism, general relativity and various quantum field theories. It is well-known however that the many consequences of this principle--such as Lagrangian mechanics, Hamilton's equations, and Noether's theorem-- apply only to conservative systems and can not capture the irreversible effects of a general dissipative system, such as the diffusion of heat. 
\end{flushleft}

\begin{flushleft}
\justify 
The purpose of this work is two-fold. First, we construct an action functional whose stationary points satisfies the (nonlinear) conservation of energy equation with heat dissipation in one space dimension. For the first time, to our knowledge, a stationary principle for the first law of thermodynamics analogous to Hamilton's principle of stationary action is formulated in this paper. This will allow us, among other things, to find a bona fide variational principle for the coupled heat equation in classical thermoelasticity and the classic heat equation. Second, we show that this procedure is iterative, which allows for the construction of an infinite number of variational principles. 
\end{flushleft}

\begin{flushleft}
\justify 
Consider one of the most basic equations in all of physics: the classic wave equation in some domain $\mathcal{B} \subseteq \mathbb{R}^n$
\begin{equation} \label{1.1}
\hspace*{16mm} \dfrac{\partial^2 u}{\partial t^2} -c^2 \nabla^2 u = 0 \ , \hspace*{7mm} (\vec{x},t) \in  \mathcal{B} \times (0, \infty) 
\end{equation}
where the scalar $c$ denotes the speed of propagation of the wave. The Lagrangian $\mathcal{L}$ for the above PDE is the difference between the kinetic energy and the (potential) strain energy
\begin{equation*} \label{1.2}
\mathcal{L} = \dfrac{1}{2} \int_{\mathcal{B}} \left( \dfrac{\partial u}{\partial t} \right)^2 - c^2 |\nabla u | ^2 d \mathcal{B} 
\end{equation*}
Then formally by Hamilton's principle of least action 
$$ \delta \int_{0}^{\tau} \mathcal{L} dt= 0 $$
we can recover the wave equation. In other words, the solution $u$ to \eqref{1.1} in $[0,\tau]$ corresponds to the stationary points of $\int_{0}^{\tau} \mathcal{L} dt$. Once $\mathcal{L}$ is defined we can rewrite the wave equation in Euler-Lagrange form
\begin{equation*} \label{1.3}
\dfrac{\partial }{\partial t}  \left( \dfrac{\delta \mathcal{L}}{\delta \dot{u}} \right) - \dfrac{\delta \mathcal{L}}{\delta u} = 0 
\end{equation*}
If function $u$ has compact support on $\mathcal{B}$ (or decays sufficiently fast when $\mathcal{B} = \mathbb{R}^n$) then the total energy of the system is conserved: 
\begin{equation*} \label{1.4}
\partial_t \mathcal{H} = 0 
\end{equation*}
where the total energy is given as
\begin{equation*} \label{1.5}
\mathcal{H} = \dfrac{1}{2} \int_{\mathcal{B}} \left( \dfrac{\partial u}{\partial t} \right)^2 + c^2 |\nabla u | ^2 d \mathcal{B} 
\end{equation*}
\end{flushleft}

\begin{flushleft}
\justify 
These classical results, however, do not have a counterpart for the classic heat equation
\begin{equation} \label{1.6}
 \dfrac{\partial u}{\partial t}  -\alpha \nabla^2 u = 0 
\end{equation}
where constant $\alpha$ is the  thermal diffusivity, and the function $u$ here represents the temperature field. In fact, it is shown \cite{berdichevsky2009variational} that no action exists in the form of 
\begin{equation} \label{1.7}
\int_{\mathcal{B}} L( \partial_t u, \nabla u, x, t) d\mathcal{B}
 \end{equation}
such that \eqref{1.6} can be deduced as the Euler-Lagrange equations of functional having the form give in equation \eqref{1.7}. 
\end{flushleft}

\begin{flushleft}
\justify 
As far as we can tell, the first successful attempt to include dissipative effects into the classsic variational framework dates back to Rayleigh in the end of the 19th century \cite{rayleigh1945theory}. Rayleigh introduced, in addition to the Lagrangain, a \emph{dissipation function}--a positive quadratic function in the velocities--to account for friction in the system; this allowed for the extension of Lagrange's equations of motion. In Onsager's groundbreaking work \cite{onsager1931reciprocal} on linear irreversible thermodynamics, one of the first variational  principles for irreversible thermodynamics was formulated: the principle of the least dissipation of energy. This principle applied to a heat conducting solid produces the steady states for the temperature distribution. In the 1950s M.A. Biot developed a variational formulation for the equations of classical thermoelasticity by means of a modified free energy (referred to as Biot’s potential) and a dissipation function. However, Biot's variational formulation was not given in terms of a single action such as in Hamilton's principle; rather a quasi-variational principle was formulated in-which the total variation of the dissipation function is not considered--only the  product of its derivatives with the appropriate infinitesimal variation \cite{biot1956thermoelasticity, biot1970variational}. 
\end{flushleft}

\begin{flushleft}
\justify 
Since then many have sought to uncover variaitonal formulations for dissipative continua. For instance in certain cases \cite{ottinger2005beyond, kaufman1984dissipative, morrison1986paradigm, said2019lagrangian} researchers were successful in extending the classic Lagrangian and/or Hamiltonian formulations to include effects of entropy production, while falling-short of constructing a unified action. Others \cite{gurtin1964variational, yang2006variational} were able to formulate action integrals for evolution-diffusion equations. However, these functionals do not have the simple form of a density function as does the Lagrangian, rather they are complicated expressions given in terms of convolutions of a one parameter integral, and it is not obvious how they would fit a Lagrangian-Hamiltonian framework. Dissipation has also been incorporated into various variational schemes via the Lagrange-d’Alembert principle (see for example \cite{kane2000variational, bloch1996euler, gay2017lagrangian}). For other possible extensions the reader is referred to \cite{PhysRevLett.110.174301, riewe1997mechanics}. Moreover, a Noether's theorem was advanced in \cite{kalpakides2004canonical} for theory of nonlinear thermoelasticity without dissipation. But despite of progress in this area, a unified extension of the variational formalism of analytical mechanics to general dissipative systems remains still out of reach.
\end{flushleft}

\begin{flushleft}
\justify 
In this paper we construct a new least action principle, albeit in one space dimension, analogous to Hamilton's principle by calculating the rate of change in the energy flux. This allows us to write the conservation of energy equation as a second order hyperbolic PDE for the total energy of the system in one space dimension. A myriad of consequence will then follow. The hyperbolic PDE can be rewritten as the Euler-Lagrange equations of a new action we denote $\Sigma$. This produces a natural way of revealing the symmetry that exists between the balance of energy and momentum, and as such the new least action principle follows without extraneous physics assumptions. Hamiltonian and bracket formalisms also follow from the Euler-Lagrange equations in a similar fashion to analytical mechanics. Furthermore, the symmetries leaving $\Sigma$ invariant correspond to new conservation laws. 
\end{flushleft}

\begin{flushleft}
\justify 
A noteworthy consequence of the above procedure is that it gives ground for producing a third functional, this time from the energy equation associated with $\Sigma$ (i.e. Noether's theorem under time invariance applied to $\Sigma$). In fact, we can carry this procedure indefinitely, that is we prove that our procedure is iterative giving rise to a hierarchy of variational principles each  constructed from the previous iteration. We, hence, obtain an infinite number of functionals (i.e. Lagrangians) and an (infinite) iterative scheme, and advance a complete analytic mechanics at each iteration.
\end{flushleft}

\begin{flushleft}
\justify 
The paper is organized as follows. We consider isentropic systems in Section 2--5, which will set the foundation for tackling the dissipative case and constructing the	variational hierarchy. We begin in Section 2 by considering the rate of change in the energy flux. We show that the total energy propagates according to a homogeneous wave equation, which can be derived as stationary points of a new functional $\Sigma$. This also gives rise to an integral of motion associated with $\Sigma$. In Section 3, we establish a unified Lagrangian-Hamiltonian formulation through the newly construed functional $\Sigma$ and its Legendre transform $\Pi$, respectively. Hamilton's equations for the total energy are shown to be equivalent to a bracket formulation over a new phase space that consists of the energy-power pair. We establish in Section 4 Noether's theorem: transformations that leave functional $\Sigma$ invariant result in conservation laws. As such we examine two groups of transformations: arbitrary translations in the total energy, and space-time translations. Invariance under the first corresponds to the balance of energy. The second is the basis for constructing a second order tensor comparable to the energy-momentum tensor in classical field theory. Subsequently we obtain a number of new conservation laws, one of which governs the evolution of density function $\pi$. One of our main results is contained in Section 5: the formulation of infinite hierarchy of variational principles (i.e. least action principles) and constants of motion. This result is made possible because the basic procedure used in constructing $\Sigma$ is iterative in nature; the Lagrangian at each iteration is formulated based on the field variables of the previous iteration. Therefore, at each iteration we produce a complete variational analysis for the relevant fields.
\end{flushleft}

\begin{flushleft}
\justify 
In Section 6 we look to apply our procedure to non-conservative systems, namely to dissipative thermoelastic material. The effects of entropy production due to heat flow manifest itself as a non-homogeneous term in the hyperbolic PDE for the total energy. Hence, a set of two coupled Euler-Lagrange equations determines the complete evolution of thermoelastic materials. Specifically, the functional $\Sigma$ for a linear thermoelastic body produces the classic energy equation in the theory of thermoelasticity, and as a special case the heat equation. Moreover, a modified Noether's theorem is given for the dissipative case. While we can not produce conservation laws--because of the inhomogeneity in the balance of energy--we obtain auxiliary equations which are fundamental for obtaining the next iteration in the variational hierarchy for the dissipative case. In Section 7, we compare our results with the well-established variational principle in the theory of irreversible thermodynamics. Finally, we offer concluding remarks in Section 8.
\end{flushleft}

\section{Least Action Principle for Balance of Energy}

\begin{flushleft}
\justify
In this section we show that the balance of energy equation in 1D corresponds to the stationary points of a functional $\Sigma$. We commence our construction by first considering the isentropic problem; while simple it will form the basis for extending $\Sigma$ to the dissipative case and constructing the infinite variational hierarchy.
\end{flushleft}

\begin{flushleft}
\justify
The point of departure for us is to consider the first law of thermodynamics (i.e. the balance of energy) in material coordinates
\begin{equation*} 
\rho_0 \dfrac{\partial I}{\partial t} = \dfrac{\partial }{\partial x_j} \left( v_iS_{ij} - q_j \right) \qquad \text{in} \; \mathcal{B} \times [0,\tau]
\end{equation*}
Here, the $n$-dimensional domain $\mathcal{B}$ represents the reference configuration ($i,j = 1, 2, ... n$), and time $\tau > 0$. The quantity $\rho_0$ is the material density of $\mathcal{B}$; $I = e + \frac{v_i v_i}{2}$ is the total energy: the sum of the internal energy and the kinetic energy; $S_{ij}$ is the first Piola-Kirchhoff stress tensor; and $q_j$ represents the heat flux across the boundary of $\mathcal{B}$. For simplicity, we have considered the balance of energy in the absence of body forces and heat sources (see Appendix A). 

\end{flushleft}

\begin{flushleft}
\justify

However, since the construction of $\Sigma$ holds in one space variable, that is $n=1, x \doteq x_1$, the above equation simplifies to
\begin{equation}  \label{2.1}
\rho_0 \dfrac{\partial I}{\partial t} = \dfrac{\partial }{\partial x} \left( v S - q \right) \qquad \text{in} \; \mathcal{B} \times [0,\tau]
\end{equation}
where now $\mathcal{B} = (a_1, a_2)$ is the reference configuration, and the quantities $v_i$, $S_{ij}$ and $q_j$ reduce to scalars $v$, $S$ and $q$, respectively.

\end{flushleft}

\begin{flushleft}
\justify
We further assume that our constitutive laws determine a classic thermoelastic medium
\begin{equation} \label{2.2}
e = e( \partial_x u,s), \quad  S = \rho_0 \dfrac{\partial e}{\partial (\partial_x u)}, \quad  \theta = \dfrac{\partial e}{\partial s}, \quad q = - k \, \dfrac{\partial \theta}{\partial x} 
\end{equation}
Here $s$ is the entropy density, and $\theta$ is the absolute temperature (i.e. $\theta = \theta_0 + T$ where $\theta_0$ is constant temperature in the undeformed configuration). Equation \eqref{2.1} holds for solids as well as fluids. Throughout this paper we assume from the outset that the heat flux $q$ is modeled by Fourier's law of heat conduction $q = - k \, \frac{\partial \theta}{\partial x} $ , where $k$ here is the thermal conductivity. In case of fluids the pressure is given as $S = -p$, and the internal energy $e$ (hence the pressure) is a state function of the density $\rho$ associated with the motion through the relation $\dfrac{\rho_0}{1+\partial_x u} =\rho$. We emphasize that the constitutive relations \eqref{2.2} define a (possibly nonlinear) thermoelastic medium. Viscoelastic material for example, where the stress depends of the strain rate, is not included.
\end{flushleft}

\begin{flushleft}
\justify
It is known that hyperbolic PDEs such as the wave equation posses a variational structure because they are associated with symmetric operators \cite{finlayson2013method}. The balance of energy equation on the other hand, supplemented by Fourier's law for heat conduction, usually results in a nonlinear parabolic PDE in the temperature field\footnote{A departure from the classic constitutive laws can result in a wave-like equation for the temperature field. See for example \cite{green1993thermoelasticity, straughan2011heat}.}. In fact, even upon linearization, this equation does not posess an obvious variational form since it will ultimately correspond to a non-symmetric operator \cite{yang2006variational}. We circumvent this impediment by showing that the total energy $I$ satisfies a wave equation.
\end{flushleft}

\begin{flushleft}
\justify
We begin by considering the isentropic problem: $q=0, e=e(\partial_x u), S=S(\partial_x u), \partial_t s =0$. The rate of change of the energy flux $F = vS$ reads
$$ \dfrac{\partial F}{\partial t} = S\dfrac{\partial v}{\partial t} + v \dfrac{\partial S}{\partial t}$$
We substitute for the conservation of momentum equation (in the absence of body forces) to obtain
\begin{align*} 
\dfrac{ \partial  F}{\partial t} =& S \dfrac{1}{\rho_0} \dfrac{ \partial S}{\partial x} + v \dfrac{\partial S}{\partial t} \\
=& \dfrac{S}{\rho_0} \dfrac{\partial S}{\partial (\partial_x u)} \dfrac{\partial (\partial_x u)}{\partial x }  + v \dfrac{\partial S }{\partial (\partial_x u)} \dfrac{\partial (\partial_x u)}{\partial t}  \\
=&  \dfrac{\partial S }{\partial (\partial_x u)} \left( \dfrac{\partial e}{\partial (\partial_x u)} \dfrac{\partial(\partial_x u)}{\partial x} + v \dfrac{\partial v}{\partial x} \right) \\
=& \dfrac{\partial S }{\partial (\partial_x u)} \dfrac{\partial }{\partial x} \left( e + \dfrac{v^2}{2} \right) \\
=& c \dfrac{\partial I}{\partial x} \numberthis \label{2.3}
\end{align*}
where 
$$c \doteq \dfrac{\partial S }{\partial (\partial_x u)}$$
is the nonlinear elastic modulus.
\end{flushleft}

\begin{flushleft}
\justify
Now Equation \eqref{2.3} together with \eqref{2.1} produces
\begin{equation}  \label{2.4}
\rho_0 \dfrac{\partial^2I}{\partial t^2} - \dfrac{\partial }{\partial x} \left( c\dfrac{\partial  I}{\partial x} \right) =0
\end{equation}
Therefore for constant $c$
\begin{equation}  \label{2.5}
\dfrac{\partial^2I}{\partial t^2} - \dfrac{c}{\rho_0}\dfrac{\partial ^2 I}{\partial x^2} =0
\end{equation}
We have obtained the classic wave equation for the total energy $I$, which we choose to replace the classic balance of energy equation \eqref{2.1} when $q=0$. The speed of propagation is well-defined and finite (for constant $c$) and is equal to the speed of propagation of an elastic wave $\sqrt{c/ \rho_0}$. That the total energy is governed by a wave equation should not come as a surprise. The propagating elastic wave packet has an associated energy, which propagates with the wave into the material. So as long as $ u \neq 0$ the associated total energy propagates with the wave in the same direction.
\end{flushleft}

\begin{flushleft}
\justify
Two immediate consequences present themselves. The first is that both equations \eqref{2.4} and \eqref{2.5}--that is, whether $c$ is constant or not--can be derived as the stationary points of an action functional, which we denote $\Sigma$. The functional $\Sigma$ is defined analogously to the classical Lagrangian in terms of $I$
\begin{equation} \label{2.6}
\Sigma(I) = \int_{\mathcal{B}} \sigma (\partial_t I, \partial_x I) \, dx \doteq \int_{\mathcal{B}} \dfrac{1}{2} \left( \rho_0 \left( \dfrac{\partial I}{\partial t} \right)^2 - c \left( \dfrac{\partial I}{\partial x} \right)^2 \right) \, dx
\end{equation}
Therefore have obtained a least action principle analogous to Hamilton's principle
\begin{theorem}[A least action principle] \label{th1}
The actual evolution of the total energy $I$ in $[0,\tau]$ coincides with the stationary points of the functional
\begin{equation*} \label{2.7}
\int_{0}^{\tau} \Sigma \, dt 
\end{equation*} 
\end{theorem}

\end{flushleft}

\begin{remark}
\justify
\leavevmode
\begin{enumerate} [i]
\item The symmetric operator associated with \eqref{2.4} or \eqref{2.5} (in weak form) can be constructed as follows. Let $V$ be some appropriate Hilbert space such that $I \in V$ solves \eqref{2.4} or \eqref{2.5} (e.g. $V= L^2 \left( 0, \tau; H^1_0(\mathcal{B}) \right)$), and define the operator $A: V \times V \longrightarrow \mathbb{R}$

$$ A(I_1,I_2) = \int_0^{\tau}  \int_{\mathcal{B}} \rho_0 \dfrac{\partial I_1}{\partial t} \dfrac{\partial I_2}{\partial t} \, dx dt - \int_0^{\tau}  \int_{\mathcal{B}} c \dfrac{\partial I_1}{\partial x} \dfrac{\partial I_2}{\partial x} \, dx dt $$

with $\partial_t I_1 (\cdot,0) =\partial_t I_2(\cdot,\tau) = 0$ for all $I_1, I_2 \in V$. Clearly, operator $A$ is symmetric. 

\item While equation \eqref{2.4} (or \eqref{2.5}) is to be expected, its significance relies in showing that the total energy I is the natural field to consider if one is looking to formulate a least action principle. In fact, equation \eqref{2.4} will form the essential ingredient for extending the novel least action principle to the dissipative case.

\item The calculation leading to \eqref{2.3} demonstrates the difficulty of arriving at \eqref{2.4} in 3D (or even 2D): the use of the chain rule to calculate the rate of change of the energy flux $\partial_t (v_i S_{ij})$ in higher dimensions will result in extra terms in \eqref{2.3} not involving the total energy $I$.

\item From a calculus of variation perspective, the variation in the total energy in \eqref{2.6} were assumed to be compact in space and time. Compactness in time only result in an additional term in \eqref{2.6} which express the force exerted at the end points \cite{marsden1994mathematical}. 
\end{enumerate}

\end{remark}
\begin{flushleft}
\justify
Another consequence of \eqref{2.5} is that we can obtain an analogous result to conservation of energy in wave mechanics
\begin{corollary}[Integral of motion] \label{cor1}
If $I$ satisfies \eqref{2.5} and has compact support on $\mathcal{B}$. Then
$$\Pi (t) =  \int_{\mathcal{B}} \dfrac{1}{2} \left( \rho_0 \left( \dfrac{\partial I}{\partial t} \right)^2 + c \left( \dfrac{\partial I}{\partial x} \right)^2 \right) \, dx$$
is an integral of motion, that is $\Pi(t) = \Pi(0)$.
\end{corollary}
\end{flushleft}

\begin{remark}
\justify
\leavevmode

As we have seen, Theorem 1 is valid for non-constant $c$. In fact, the complete analytical mechanics formulation can be extended to the case of non-constant $c$ as well as to a more general formulation of the first law (see Appendix A for details).
\end{remark}

\section{Lagrangian-Hamiltonian Formulation}
\begin{flushleft}
\justify
In this section we formulate the Lagrangian-Hamiltonian theory corresponding to the balance of energy, now rewritten in hyperbolic form  \eqref{2.4} or \eqref{2.5}.
\end{flushleft}

\begin{flushleft}
\justify
By Theorem \ref{th1}, the Euler-Lagrange equations for the total energy $I$ directly follow in terms of the functional $\Sigma$ 
\begin{equation} \label{3.1}
\dfrac{d }{d t} \left( \dfrac{\delta \Sigma}{\delta (\partial_t I)} \right) - \dfrac{\delta \Sigma}{\delta I } =0
\end{equation}
where $\dfrac{\delta }{\delta u}$ denotes the functional derivative. Clearly, the Euler-Lagrange equation is equivalent to \eqref{2.4}  (or to \eqref{2.5} for constant $c$).
\end{flushleft}

\begin{flushleft}
\justify 
The Hamiltonian formulation corresponding to the new least action principle follows analogously to classical mechanics. We first introduce the variable $J$ defined analogously to the momentum variable in mechanics
\begin{equation} \label{3.2}
J = \dfrac{\delta \Sigma}{\delta (\partial_t I)}
\end{equation}
\end{flushleft}

\begin{flushleft}
\justify 
By taking the Legendre transformation of the function $\sigma$, the density function of $\Sigma$, with respect to the change of variable $(I, \partial_t I) \longrightarrow (I, J)$, we obtain an analogous quantity to the Hamiltonian density
\begin{equation} \label{3.3}
\pi(\nabla I,J) \doteq J \cdot\dfrac{\partial I}{\partial t} - \sigma
\end{equation}
A simple calculation gives the total quantity of $\pi$ inside $\mathcal{B}$
\begin{equation*} \label{3.4}
\Pi \doteq \int_{\mathcal{B}} \pi d \mathcal{B} = \int_{\mathcal{B}} \left( \dfrac{1}{2\rho_0} J^2 +\dfrac{c}{2} \left( \dfrac{\partial I}{\partial x} \right)^2 \right) dx
\end{equation*}
\end{flushleft}

\begin{flushleft}
\justify 
The physical interpretation of $J$ becomes evident once we carry-out the calculations in \eqref{3.1}
$$ J = \dfrac{\delta \Sigma}{\delta (\partial_t I)} = \rho_0 \dfrac{\partial I}{\partial t} $$
The quantity $J$ is the total power density of the system. Therefore, the phase space associated with $\Pi$ specifies the energy-power pair of the system at each instant time. 
\end{flushleft}

\begin{flushleft}
\justify 
Equipped with he Euler-Lagrange equation \eqref{3.1} together with  functional $\Pi$, we can rewrite the nonlinear evolution equation for the total energy $I$ as Hamilton's equations in $\Pi$ 

\begin{align*}
\dfrac{\partial I}{\partial t} = \dfrac{\delta \Pi}{\delta J}  \\
\dfrac{\partial J}{\partial t}   = - \dfrac{\delta \Pi}{\delta I} 
\end{align*}
\end{flushleft}

\begin{flushleft}
\justify 
For arbitrary functionals $\Phi_1$ and $\Phi_2$ over the (infinite dimensional) phase space associated with the functional $\Pi$: $\lbrace (I,J) \mid I  = e +\frac{v^2}{2} \ , J = \frac{\delta \Sigma}{\delta (\partial_t I)} \rbrace$, we can introduce the canonical Poisson bracket $ \lbrace , \rbrace$
\begin{equation} \label{3.6}
\lbrace \Phi_1, \Phi_2 \rbrace  = \int_{\mathcal{B}} \dfrac{\delta \Phi_1}{\delta I} \dfrac{\delta \Phi_2}{\delta J} - \dfrac{\delta \Phi_2}{\delta I}  \dfrac{\delta \Phi_1}{\delta J} \, d\mathcal{B}
\end{equation}
The evolution of an arbitrary functional $\Phi$ over the phase space is governed by $\Pi$ through the Poisson structure
$$ \dot{\Phi}(I,J) = \int_{\mathcal{B}} \left( \dfrac{\delta \Phi}{\delta I}, \dfrac{\delta \Phi}{\delta J} \right) \cdot  \left( \dfrac{\partial I}{\partial t}, \dfrac{\partial J}{\partial t} \right)^T d\mathcal{B}  = \int_{\mathcal{B}} \dfrac{\delta \Phi}{\delta I} \dfrac{\delta \Pi }{\delta J} -  \dfrac{\delta \Phi}{\delta J} \dfrac{\delta \Pi }{\delta I} d\mathcal{B} = \lbrace \Phi, \Pi \rbrace $$
Specifically, we can rewrite Hamilton's equations \eqref{3.3} in Poisson form 

\begin{align*}
\dfrac{\partial I}{\partial t} = \lbrace I, \Pi \rbrace   \\
\dfrac{\partial J}{\partial t}  = \lbrace J, \Pi \rbrace 
\end{align*}

Finally, if we assume that $\Pi$ is strictly a function over the phase space, that is $c$ is constant, we reproduce the result of Corollary \ref{cor1}
$$ \dfrac{\partial \Pi}{\partial t}  = \lbrace \Pi, \Pi \rbrace = 0 $$
\end{flushleft}

\section{Variational Symmetries and Conservation Laws}

\begin{flushleft}
\justify
In this section we demonstrate that invariance of the quantity $\Sigma$ under group transformations (up to a full divergence of a field) result in conservation laws, one of which is the balance of energy. We shall state a version of Noether's theorem $\Sigma$ convenient for our setting, and examine the consequences; namely by constructing a quantity analogous to the energy-momentum tensor in field theory.
\end{flushleft}

\begin{flushleft}
\justify
Set $z_{\mu} = \lbrace t,x \rbrace \in \mathcal{B}_\tau \doteq [0,\tau] \times \mathcal{B}$ for arbitrary $\tau \in \mathbb{R}^+$, and $\mu =  0,1$. This notation (often used in field theory) proves more convenient in stating Noether's theorem. Under this notation, the Euler-Lagrange equation \eqref{3.1} reads
\begin{equation} \label{4.1}
\dfrac{\partial }{\partial z_{\mu} } \left(  \dfrac{\partial \sigma }{\partial (\partial_{\mu} I)} \right) = 0
\end{equation}
We further assume in this section that the elastic modulus $c$ is constant (see Remark 2).
\end{flushleft}

\subsection*{Noether's Theorem }

\begin{flushleft}
\justify 
We introduce a one-parameter smooth transformation $ \lambda \in [0, \infty) \longmapsto I_{\lambda} = I(z_{\mu}; \lambda)$ with $I_\lambda \mid_{\lambda=0} = I$. Similarly, we define $\sigma_\lambda \doteq \sigma (\partial_{\mu} I_\lambda)$. 
\end{flushleft}

\begin{flushleft}
\justify 
We say that function $\sigma$ is invariant under the one-parameter group of transformations $\lambda \longmapsto I_\lambda$ if 
\begin{equation} \label{4.2}
\dfrac{d}{d\lambda} \int_{\mathcal{B}_T} \sigma_\lambda \mid_{\lambda=0} d^2 z= \int_{\mathcal{B}_T} \dfrac{\partial K_{\mu}}{\partial z_{\mu}} d^2 z
\end{equation}
for some (possibly zero) four-vector field $\vec{K} = \vec{K}(I,\partial_{\mu} I)$. 
\end{flushleft}

\begin{theorem}[Noether's theorem for $\sigma$] \label{th2}
Assume $\sigma$ is invariant under the one-parameter group of transformations $\lambda \longmapsto I_\lambda$, then the Euler-Lagrange system corresponding to $\sigma$ admits the conservation law
\begin{equation} \label{4.3}
\dfrac{\partial P_{\mu}}{\partial z_{\mu}} = 0
\end{equation}
where the conserved current $P_{\mu}$ is defined as
\begin{equation} \label{4.4}
P_{\mu} = \dfrac{\partial \sigma}{\partial (\partial_{\mu} I)} \cdot \dfrac{\partial I_{\lambda}}{\partial \lambda} \mid_{\lambda=0} - K_{\mu} 
\end{equation}

\end{theorem}

\begin{proof}

\begin{flushleft}
\justify

We begin by computing the LHS of \eqref{4.2}

\begin{align*} 
\dfrac{d}{d\lambda} \int_{\mathcal{B}_T} \sigma_\lambda \mid_{\lambda=0} d^2 z =& \int_{\mathcal{B}_T} \dfrac{\partial \sigma }{\partial (\partial_{\mu} I)} \cdot \dfrac{\partial (\partial_{\mu} I_\lambda)}{\partial \lambda} \mid_{\lambda=0} d^2 z \\ 
=& - \int_{\mathcal{B}_T} \dfrac{\partial }{\partial z_{\mu} } \left(  \dfrac{\partial \sigma }{\partial (\partial_{\mu} I)} \right) \cdot \dfrac{\partial I_\lambda}{\partial \lambda } \mid_{\lambda=0} d^2 z + \int_{\mathcal{B}_T}  \dfrac{\partial  } {\partial z_{\mu}} \left(  \dfrac{\partial \sigma }{\partial (\partial_{\mu} I)} \cdot \dfrac{\partial I_\lambda}{\partial \lambda } \mid_{\lambda=0}  \right) d^2 z \\
=&  \int_{\mathcal{B}_T}  \dfrac{\partial  } {\partial z_{\mu}} \left(  \dfrac{\partial \sigma }{\partial (\partial_{\mu} I)} \cdot \dfrac{\partial I_\lambda}{\partial \lambda } \mid_{\lambda=0}  \right) d^2 z  \numberthis \label{4.5} \\
\end{align*}	
where we have used the fact the $I$ satisfies the Euler-Lagrange equations \eqref{4.1}. By comparing equations \eqref{4.2} and \eqref{4.5}, we obtain the conservation law for the current $P_{\mu}$. 
\end{flushleft}

\end{proof}

\begin{remark}
\leavevmode
\justify
\begin{enumerate} [i]
\item In terms of coordinates $(t,x)$, the conservation law \eqref{4.3} can be rewritten as
\begin{equation} \label{4.6}
\dfrac{\partial P_0}{\partial t} + \dfrac{\partial P_1}{\partial x} = 0
\end{equation}

\item For simplicity we have chosen $\lambda$ to be a scalar parameter. Theroem 1 is equally valid for a $\mu-$dimensional parameter $\vec{\lambda}$ \cite{ kalpakides2004canonical}.

\item A more basic diffeomorphism can be defined with respect to the independent variable, that is, $\lambda \longrightarrow z_\mu^\lambda \doteq z_\mu(\lambda)$. Then the group of transformations $I_\lambda$ is defined in terms of $z_\mu^\lambda$, namely, $I_\lambda \doteq I(z_\mu^\lambda)$. Therefore, invariance of $\sigma$ with respect to the diffeomorphism $\lambda \longrightarrow z_\mu(\lambda)$ still produces the result of Theorem 1 \cite{giaquinta2011mathematical}.
\end{enumerate}
\end{remark}

\begin{flushleft}
\justify
An immediate implication of Theorem \ref{th2} is that the conservation of energy, written in hyperbolic form \eqref{2.4}, can be derived as a conservation law corresponding to invariance under \emph{energy translations}. 
\end{flushleft}

\begin{flushleft}
\justify
To see this, we define the family of energy translations $I_\lambda = I +\lambda \overline{I}  $ for some constant energy scalar $\overline{I}$. Then, under this group of transformations, we have
$$ \sigma_\lambda = \sigma ( \partial_\mu I_\lambda) = \sigma ( \partial_\mu I) = \sigma$$
Therefore, $\sigma$ is invariant under the transformation $\lambda \longrightarrow I_\lambda$; in other words, equation \eqref{4.2} is satisfied with $\vec{K} = 0$. The conserved current $P_\mu$, in this case, reads
\begin{equation*} \label{4.7}
P_\mu = \dfrac{\partial \sigma}{\partial (\partial _\mu I)} \cdot \overline{I}
\end{equation*}
and the corresponding conservation law holds
\begin{equation*} \label{4.8}
\dfrac{\partial }{\partial z_{\mu} } \left(  \dfrac{\partial \sigma }{\partial (\partial_{\mu} I)} \right) = 0
\end{equation*}
which is nothing other than the balance of energy equation written in Euler-Lagrange form \eqref{4.1}. 
\end{flushleft}

\begin{flushleft}
\justify
Another consequence of Theorem \ref{th2} is the global conservation law for \emph{total charge} $Q $. Integrating both sides of \eqref{4.6} over the some large region (interval) $\Omega $ yields:
$$ \int_{\Omega} \dfrac{\partial P_0}{\partial t} d\Omega= - \int_{\Omega} \dfrac{\partial P_1}{\partial x} d\Omega = - P_1 \mid_{\partial \Omega}$$
By assuming $P_1$ has compact support on $\Omega$, we conclude that the total charge $Q \doteq \int_{\Omega} P_0 d\Omega$ is conserved
\begin{equation*} \label{4.9}
\dfrac{d}{dt}Q = 0
\end{equation*}
\end{flushleft}

\subsection*{Energy-Momentum Tensor}

\begin{flushleft}
\justify
Among the conservation laws associated with a field Lagrangian, those that are derived from the \emph{energy-momentum tensor} are the most  significant from a physics perspective. Here again we can construct a quantity analogous to that of classical field theory, namely the energy-momentum tensor, and obtain a host of conservation laws as a result of Theorem \ref{th2}. 
\end{flushleft}

\begin{flushleft}
\justify
We begin the construction by considering the following space-time translation (see Remark 3-ii)
\begin{equation} \label{4.10}
z^\lambda_\mu \doteq z_\mu - \lambda_\mu =  z_\mu - \lambda_\eta \delta_{\mu \eta}
\end{equation}
where $\mu, \eta = 0,1$.
\end{flushleft}

\begin{flushleft}
\justify
Under this perturbation the field $I$ can be written as\footnote{We shall suppress the index $\eta, \mu$ on $\lambda$ if the parameter $\lambda$ appears as a subscript.}
\begin{equation*} \label{4.11}
I_\lambda = I(z_\mu - \lambda_\mu) = I(z_\mu) + \lambda_\mu \dfrac{\partial I_{\lambda}}{\partial z_\mu} \mid_{\lambda = 0} + o(| \lambda_\mu|) 
\end{equation*}
\end{flushleft}

\begin{flushleft}
\justify
Next, we show $\sigma_{\lambda}$ satisfies \eqref{4.2}
\begin{proposition} \label{prop1}
Function $\sigma$ is invariant under the space-time translation \eqref{4.10}.
\end{proposition} 
\end{flushleft}

\begin{proof}
\justify
Function $\sigma_{\lambda}$ is defined as $\sigma(\partial_\alpha I_{\lambda})$. Expanding this expression in $\lambda_\mu$ yields
\begin{align*}
 \sigma_{\lambda}  &= \sigma(\partial_\alpha I ) + \lambda_\mu \dfrac{d}{d \lambda_\mu} \sigma_{\lambda} \mid_{\lambda = 0} +  o(| \lambda_\mu|) \\
&= \sigma + \lambda_\mu \dfrac{\partial \sigma}{\partial (\partial_\alpha I)} \dfrac{\partial (\partial_\alpha I)}{\partial z_\mu } +  o(| \lambda_\mu|)  \\
&= \sigma + \lambda_\mu \dfrac{\partial \sigma}{\partial z_\mu} + o(| \lambda_\mu|) 
\end{align*}

where $\alpha, \mu = 0,1$. Hence
\begin{equation} \label{4.13}
\dfrac{d}{d \lambda_{\eta}} \sigma_\lambda \mid_{\lambda = 0} = \dfrac{\partial }{\partial z_{\mu}} \left( \delta_{\mu \eta} \sigma \right)
\end{equation}
Since we have assumed a diffeomorphism with respect to a two dimensional parameter $\lambda_\mu$, the vector field in \eqref{4.2} is augmented to a second order tensor. Therefore, according to \eqref{4.13}, we take $K_{\mu \eta} = \delta_{\mu \eta} \sigma$ to satisfy the condition \eqref{4.2}.

\end{proof}

\begin{flushleft}
\justify
As a result of Proposition \ref{prop1} (together with Theorem \ref{th2}), we have the following system of conservation laws
\begin{equation} \label{4.14}
\dfrac{\partial T_{\mu \eta}}{\partial z_\mu} =0
\end{equation}
where the tensor $T_{\mu \eta}$, defined by
\begin{equation} \label{4.15}
T_{\mu \eta} = \dfrac{\partial \sigma}{\partial \left( \partial_{\mu} I \right)} \cdot \dfrac{\partial I}{\partial z_\eta} - \delta_{\mu \eta} \sigma
\end{equation}
is a quantity analogous to the energy-momentum tensor in classical (and quantum) field theory. The system of conservation laws \eqref{4.14} are accompanied, again, by the conservation of their global counterparts $\int_\Omega T_{00} d\Omega$, and $\int_\Omega T_{01} d\Omega$ as discussed earlier. Moreover, as one expects, 
$$ T_{00} =  \dfrac{\partial \sigma}{\partial \left( \partial_0 I \right)} \cdot \dfrac{\partial I}{\partial z_0} - \sigma = \rho_0 \dfrac{\partial I}{\partial t} \dfrac{\partial I}{\partial t} - \sigma = \pi $$
and the evolution of $\pi$ is governed by 
\begin{equation} \label{4.16}
\dfrac{\partial \pi}{\partial t} + \dfrac{\partial }{\partial x} \left(  \dfrac{\partial I}{\partial t} \dfrac{\partial \sigma}{\partial (\partial_x I)}  \right) = 0
\end{equation}
\end{flushleft}

\begin{flushleft}
\justify
Finally, we can conclude that $T_{\mu \eta}$ is symmetric by writing $\sigma$ as
\begin{equation*} \label{4.17}
\sigma = \dfrac{1}{2}B_{\mu \eta} \dfrac{\partial I}{\partial z_\mu} \dfrac{\partial I}{\partial z_\eta} 
\end{equation*}
where 
\begin{equation*} \label{4.18}
B_{\mu \eta} = \begin{bmatrix}
\rho_0 & 0 \\
0 & - c \\
\end{bmatrix}
\end{equation*}
Since $B_{\mu \eta}$ is symmetric, the tensor $T_{\mu \eta}$ in this case:
\begin{equation*} \label{4.19}
T_{\mu \eta} = B_{ \mu \alpha} \dfrac{\partial I}{\partial z_\alpha} \dfrac{\partial I}{\partial z_\eta} - \delta_{\mu \eta} \sigma
\end{equation*}
must also be symmetric. 
\end{flushleft}

\section{A Hierarchy of Variational Principles}

\begin{flushleft}
\justify
In this section we show that the process sketched thus far is iterative in nature. At each iteration, a corresponding ``Lagrangian'' can be constructed. As a result, we can formulate a least action principle at the $i$th iteration, and all the results of Sections 2--4 can be produced once again. Moreover, the constituents of each new variational principle depend on the preceding level of analysis. Therefore, we can visualize a hierarchy comprising of an infinite number of interrelated Lagrangians and their resulting variational principles. 
\end{flushleft}

\begin{flushleft}
\justify
To clearly illustrate the procedure for obtaining the general iteration, we first consider the following. In Section 2, the rate of energy flux (in the isentropic case) $\partial_t (vS)$ was computed and was shown to be proportional to the gradient of the total energy (i.e. equation \eqref{2.3}). In the same spirit, we can view the term $\partial_t I \frac{\partial \sigma}{\partial (\partial_x I)}$ as the ``energy flux'' associated with $\pi$ in \eqref{4.16}. In fact, this equation is completely analogous to the isentropic balance of energy equation
$$\rho_0 \dfrac{\partial I}{\partial t} - \dfrac{\partial }{\partial x} \left( vS \right)  = \dfrac{\partial I}{\partial t}  + \dfrac{\partial }{\partial x} \left( \dfrac{\partial u}{\partial t}\dfrac{\partial L}{\partial (\partial_x u)} \right) =0$$
Therefore, it is natural to consider the rate of change of $\partial_t I \frac{\partial \sigma}{\partial (\partial_x I)}$:
\begin{align*} 
\dfrac{\partial }{\partial t} \left(\dfrac{\partial I}{\partial t} \frac{\partial \sigma}{\partial (\partial_x I)}  \right) &= - \dfrac{\partial }{\partial t} \left( c \dfrac{\partial I}{\partial x} \dfrac{\partial I}{\partial t}  \right) \\
&= - c \dfrac{\partial^2 I}{\partial x \partial t} \dfrac{\partial I}{\partial t} - c\dfrac{\partial^2 I}{ \partial t^2} \dfrac{\partial I}{\partial x} \\
&= - c \dfrac{\partial^2 I}{\partial x \partial t} \dfrac{\partial I}{\partial t} - \dfrac{c^2}{\rho_0}\dfrac{\partial^2 I}{ \partial x^2} \dfrac{\partial I}{\partial x} \\
&= - \dfrac{c}{\rho_0} \dfrac{\partial }{\partial x} \left( \dfrac{\rho_0}{2} \left( \dfrac{\partial I}{\partial t} \right)^2 \right)  - \dfrac{c}{\rho_0} \dfrac{\partial }{\partial x} \left( \dfrac{c}{2} \left( \dfrac{\partial I}{\partial x}  \right)^2  \right) \\
&= - \dfrac{c}{\rho_0} \dfrac{\partial }{\partial x} \left(\dfrac{\rho_0}{2} \left( \dfrac{\partial I}{\partial t} \right)^2 + \dfrac{c}{2} \left( \dfrac{\partial I}{\partial x} \right)^2 \right) \\
&= -  \dfrac{c}{\rho_0} \dfrac{\partial}{\partial x} \pi \numberthis \label{5.1}
\end{align*}
Equation \eqref{5.1} together with \eqref{4.16} gives
\begin{equation} \label{5.2}
\dfrac{\partial^2 \pi}{\partial t^2} - \dfrac{c}{\rho_0}\dfrac{\partial ^2 \pi}{\partial x^2} =0
\end{equation}
We have obtained the classic wave equation for the quantity $\pi$. Hence, by treating $\pi$ as an independent quantity--as we did with the total energy $I$--all the results proven in Sections 2--4 hold too once we replace $I$ with $\pi$, and treat $\pi$ as the new field.   
\end{flushleft}

\begin{flushleft}
\justify
We, now, present the main result of this section: the preceding calculation can be put into an iterative scheme
\begin{theorem}
Let $L^0 = \dfrac{\rho_0}{2} \left( \dfrac{\partial u^0}{\partial t} \right)^2 - \dfrac{c}{2} \left( \dfrac{\partial u^0}{\partial x} \right)^2$ denote the classic Lagrangian with displacement field $u^0(x,t)=u(x,t)$. Assume $c \doteq \dfrac{\partial^2 L}{\partial (\partial_x u)^2} > 0$ is constant. Then there exist infinitely many density functionals $L^i$, $i=1,2 ... $ satisfying the least action principle
\begin{equation} \label{5.3}
\delta^i \int_0 ^\tau \int_{\mathcal{B}} L^i dx dt =0
\end{equation} 
where $L^i =  \dfrac{\rho_0}{2} \left( \dfrac{\partial u^i}{\partial t} \right)^2 - \dfrac{c}{2} \left( \dfrac{\partial u^i}{\partial x} \right)^2$, $ u^{i+1} = \dfrac{\rho_0}{2} \left( \dfrac{\partial u^i}{\partial t} \right)^2 + \dfrac{c}{2} \left( \dfrac{\partial u^i}{\partial x} \right)^2$, and $\delta^i$ is the variation taken with respect to $u^i$,  $ i=0,1,2 ... \quad .$ 
\end{theorem}
\end{flushleft}

\begin{proof}
We prove by induction.

\begin{flushleft}
\justify
For $i=0$, we obtain Hamilton's principle of least action. 
\end{flushleft}

\begin{flushleft}
\justify
Assume theorem holds for the $i$-th iteration. We look to construct the $(i+1)$th iteration. Assume that $u_i$ is a solution to the $i$th variational problem. First we compute the Euler-Lagrange equation at this iteration
\begin{equation*} \label{5.4}
 \dfrac{d}{dt} \left( \dfrac{\partial L^i}{\partial (\partial_t u^i)} \right) + \dfrac{\partial }{\partial x} \left( \dfrac{\partial L^i}{\partial (\partial_x u^i)} \right) = \rho_0 \dfrac{\partial^2 u^i}{\partial t^2} - c \dfrac{\partial^2 u^i}{\partial x^2} = 0
\end{equation*}
Next consider 
\begin{align*} 
\dfrac{\partial u^{i+1}}{\partial t} &= \rho_0 \dfrac{\partial u^i}{\partial t}  \dfrac{\partial^2 u^i}{\partial t^2} + c \dfrac{\partial u^i}{\partial x}\dfrac{\partial^2 u^i}{\partial t \partial x} \\
&= c \dfrac{\partial u^i}{\partial t}  \dfrac{\partial^2 u^i}{\partial x^2}  + c \dfrac{\partial u^i}{\partial x}\dfrac{\partial^2 u^i}{\partial t \partial x} \\
&= \dfrac{\partial }{\partial x} \left( \dfrac{\partial u^i}{\partial t} c \dfrac{\partial u^i}{\partial x} \right) \\
&= - \dfrac{\partial }{\partial x} \left( \dfrac{\partial u^i}{\partial t} \dfrac{\partial L^i}{\partial (\partial_x u^i)} \right)
\end{align*}
Thus, we have the local balance of energy
\begin{equation} \label{5.5}
\dfrac{\partial u^{i+1}}{\partial t} + \dfrac{\partial }{\partial x} \left( \dfrac{\partial u^i}{\partial t} \dfrac{\partial L^i}{\partial (\partial_x u^i)} \right) = 0
\end{equation}
\end{flushleft}
Finally, we calculate the energy flux rate
\begin{align*} 
\dfrac{\partial }{\partial t} \left( \dfrac{\partial u^i}{\partial t} \dfrac{\partial L^i}{\partial (\partial_x u^i)} \right) &= - \dfrac{\partial }{\partial t} \left( \dfrac{\partial u^i}{\partial t} c \dfrac{\partial u^i}{\partial x} \right) \\
&= - c \dfrac{\partial^2 u^i}{\partial t^2} \dfrac{\partial u^i}{\partial x} - c \dfrac{\partial u^i}{\partial t} \dfrac{\partial^2 u^i}{\partial t \partial x} \\
&= - \dfrac{c^2}{\rho_0} \dfrac{\partial^2 u^i}{\partial x^2} \dfrac{\partial u^i}{\partial x} -  c \dfrac{\partial u^i}{\partial t} \dfrac{\partial^2 u^i}{\partial t \partial x} \\
&= -c \dfrac{\partial }{\partial x} \left( \left(\dfrac{\partial u^i}{\partial t} \right)^2 \right)  -\dfrac{c^2}{\rho_0} \dfrac{\partial }{\partial x} \left( \left(\dfrac{\partial u^i}{\partial x} \right)^2 \right) \\
&= -\dfrac{c}{\rho_0} \dfrac{\partial }{\partial x} \left( \rho_0 \left(\dfrac{\partial u^i}{\partial t} \right)^2 + c \left(\dfrac{\partial u^i}{\partial x} \right)^2  \right) \\
&= -\dfrac{c}{\rho_0} \dfrac{\partial }{\partial x} u^{i+1} \numberthis \label{5.6}
\end{align*}
Combining \eqref{5.6} with \eqref{5.5} we obtain
\begin{equation*} \label{5.7}
\rho_0 \dfrac{\partial^2 u^{i+1}}{\partial t^2} - c \dfrac{\partial^2 u^{i+1}}{\partial x^2} = 0
\end{equation*}
Written differently
\begin{equation*} \label{5.8}
 \dfrac{d}{dt} \left( \dfrac{\partial L^{i+1}}{\partial (\partial_t u^{i+1})} \right) + \dfrac{\partial }{\partial x} \left( \dfrac{\partial L^{i+1}}{\partial (\partial_x u^{i+1})} \right) = 0 
\end{equation*}
Therefore $u^{i+1}$ is a stationary point for the $(i+1)$th variational problem.

\end{proof}

\begin{remark}
\justify
\leavevmode
A more general variational hierarchy theorem holds for non-constant $c$.  In Appendix A, we sketch the general method of obtaining the second iteration (i.e. $u^2 = \pi$), which then can be extended to higher iterations. 
\end{remark}

\begin{flushleft}
\justify
Since at each iteration the corresponding scalar field $u^i$ satisfies the wave equation \eqref{5.2}, we have an infinite number of integrals of motion

\begin{corollary}[Integrals of motion] \label{cor2}
If $u^i$ solves the variational problem \eqref{5.3} for $i =0,1,2,...$ and has compact support on $\mathcal{B}$. Then for every $i \in \mathbb{N}$
$$\mathcal{H}^i(t) =  \int_{\mathcal{B}} \dfrac{1}{2} \left( \rho_0 \left( \dfrac{\partial u^i}{\partial t} \right)^2 + c \left( \dfrac{\partial u^i}{\partial x} \right)^2 \right) \, dx$$
is an integral of motion, that is $\mathcal{H}^i(t) = \mathcal{H}^i(0)$.
\end{corollary}
\end{flushleft}

\begin{flushleft}
\justify
Since at each level of analysis we have a corresponding Lagrangian $L^i$, we can define the $i$-th Hamiltonian $H^i$ via the Legendre transform \eqref{3.2} by first defining the $i$-th momentum
\begin{equation*} \label{5.9}
p^i = \dfrac{\partial L^i}{\partial (\partial_t u^i)}
\end{equation*} 
Then $H^i$ will have the form
\begin{equation*} \label{5.10}
H^i(p^i,\partial_x u^i) = \dfrac{1}{2\rho_0}(p^i)^2+\dfrac{c}{2} \left( \dfrac{\partial u^i}{\partial x} \right)^2
\end{equation*}
Clearly for $i=1$, $p^1 = J$ and $H^1= \pi$ given in Section 3. 
\end{flushleft}

\begin{flushleft}
\justify
Therefore, at the $i$-th iteration, Hamilton's equations are given
\begin{align*}
\dfrac{\partial u^i}{\partial t} = \dfrac{\delta \mathcal{H}^i}{\delta p^i} \\
\dfrac{\partial p^i}{\partial t}   = - \dfrac{\delta \mathcal{H}^i}{\delta u^i} 
\end{align*}
And in a similar fashion to Section 3, we can define the Poisson brackets $ \lbrace, \rbrace^i $ for each $i$. 
\end{flushleft}

\begin{flushleft}
\justify
Finally, we touch on the key results of Section 4. We begin by rewriting the Euler-Lagrange equations at the $i$-th level in the space-time coordinates $z_\mu$:
\begin{equation*} 
\dfrac{\partial }{\partial z_{\mu} } \left(  \dfrac{\partial L^i }{\partial (\partial_{\mu} u^i)} \right) = 0
\end{equation*}
where the summation is over the subscripts (lower indices) only.
\end{flushleft}

\begin{flushleft}
\justify
For each $i$ fixed we obtain a conserved current $P_\mu^i$ associated with $L^i$ by Noether's theorem (Theorem 2): 
\begin{equation*} \label{5.13}
\dfrac{\partial P_\mu^i}{ \partial z_\mu } = 0
\end{equation*}
where
\begin{equation*} \label{5.14}
P_{\mu}^i = \dfrac{\partial L^i}{\partial (\partial_{\mu} u^i)} \cdot \dfrac{\partial u^i_{\lambda}}{\partial \lambda} \mid_{\lambda=0} - K_{\mu}^i
\end{equation*}
for $i=0,1,2 ... \, .$
\end{flushleft}

\begin{flushleft}
\justify
Particularly, if we consider the one-parameter group of transformation \eqref{4.10}, we can show, as done in Section 4, that $L^i$ is invariant under space-time translation:
\begin{equation*} \label{5.15}
\dfrac{d}{d \lambda_{\eta}} L^i_\lambda \mid_{\lambda = 0} = \dfrac{\partial }{\partial z_{\mu}} \left( \delta_{\mu \eta} L^i \right)
\end{equation*}
Therefore, the $i$-th energy-momentum tensor:
\begin{equation*} \label{5.16}
T_{\mu \eta}^i = \dfrac{\partial L^i}{\partial \left( \partial_{\mu} u^i \right)} \cdot \dfrac{\partial u^i}{\partial z_\eta} - \delta_{\mu \eta} L^i
\end{equation*}
satisfies
\begin{equation*} \label{5.17}
\dfrac{\partial T^i_{\mu \eta}}{\partial z_\mu} = 0
\end{equation*}
for each $i=0,1,2 ... \, .$
\end{flushleft}

\begin{flushleft}
\justify
So in addition to the integrals of motion obtained in Corollary 2, we have a host of other conservation laws (both global and local) at each iteration emanating from Noether's theorem. Hence, in aggregate, our hierarchy contains an infinite number of conservation laws. 
\end{flushleft}

\section{Application to Dissipative Systems}

\begin{flushleft}
\justify
We now turn our attention to the original problem of heat flow given by the first law of thermodynamics \eqref{2.1}. Our main goal here is to show that the total energy field $I$ solves a hyperbolic PDE as well; specifically a non-homogeneous wave equation in the total energy. The dissipative effects appear as additional terms independent of the field $I$. A more general action functional can then be constructed that includes dissipation. Therefore, a least action principle, Lagrangian-Hamiltonian formalism, and a (modified) Noether's theorem will all follow as well in the dissipative case. 
\end{flushleft}

\begin{flushleft}
\justify
We proceed as in Section 2 and consider the rate of change in the total energy flux: $G = vS -q$:
\begin{align*}
\dfrac{\partial G}{\partial t} &=S \dfrac{\partial v }{\partial t}+ v \dfrac{\partial S} {\partial t} - \dfrac{\partial q }{\partial t} \\
&= \dfrac{1}{\rho_0} S  \dfrac{\partial S } { \partial x} +  v \dfrac{\partial S}{\partial t}- \dfrac{\partial q }{\partial t}   \numberthis \label{6.1}
\end{align*}
Since the internal energy and the stress are functions of both $\partial_x u$ and $s$, we obtain additional terms in \eqref{6.1} involving the entropy $s$ in addition to those obtained in \eqref{2.3}
\begin{align*}
\dfrac{\partial G}{\partial t}  &= \dfrac{S}{\rho_0}  \dfrac{\partial S}{\partial (\partial_x u)} \dfrac{\partial^2 u }{\partial x^2 }  + v  \dfrac{\partial S}{\partial (\partial_x u)} \dfrac{\partial u^2 }{\partial t \partial x}   + \dfrac{S}{\rho_0}  \dfrac{\partial S}{\partial s} \dfrac{\partial s}{\partial x } +  v \dfrac{\partial S}{\partial s} \dfrac{\partial s}{\partial t} - \dfrac{\partial q }{\partial t} \\
&= \dfrac{\partial S}{\partial (\partial_x u)} \left( \dfrac{\partial e}{\partial (\partial_x u)} \dfrac{\partial^2 u }{\partial x^2 } + v \dfrac{\partial v}{\partial x} \right) + \dfrac{\partial S}{\partial s} \left( \dfrac{\partial e}{\partial (\partial_x u) } \dfrac{\partial s}{\partial x} + v \dfrac{\partial s}{\partial t} \right) - \dfrac{\partial q }{\partial t} \\
&=  \dfrac{\partial S}{\partial (\partial_x u)} \left( \dfrac{\partial e}{\partial x } +  \dfrac{\partial }{\partial x} \left( \frac{1}{2} v^2 \right) \right) + \dfrac{\partial S}{\partial s} \left( \dfrac{\partial e}{\partial (\partial_x u) } \dfrac{\partial s}{\partial x} + v \dfrac{\partial s}{\partial t} \right) - \dfrac{\partial S}{\partial (\partial_x u)} \dfrac{\partial e}{\partial s} \dfrac{\partial s}{\partial x} - \dfrac{\partial q }{\partial t} \\
&= c \dfrac{\partial I}{\partial x} + \dfrac{\partial S}{\partial s} v \dfrac{\partial s}{\partial t} + \dfrac{\partial s}{\partial x} S^2 \dfrac{\partial }{\partial (\partial_x u)} \left( \dfrac{\theta}{S} \right)- \dfrac{\partial q }{\partial t} \\
&= c \dfrac{\partial I}{\partial x} + \dot{D}  \numberthis \label{6.2}
\end{align*}
where we have defined
\begin{equation*} \label{6.3}
\dot{D} = \dfrac{\partial S}{\partial s} v \dfrac{\partial s}{\partial t} + \dfrac{\partial s}{\partial x} S^2 \dfrac{\partial }{\partial (\partial_x u)} \left( \dfrac{\theta}{S} \right)- \dfrac{\partial q }{\partial t}
\end{equation*}
Clearly, for a conservative system $\dot{D}$ is identically zero and we recover \eqref{2.3} as expected. 
\end{flushleft}

\begin{flushleft}
\justify
Now by equation \eqref{2.1} and \eqref{6.2} we obtain
\begin{equation} \label{6.4}
\rho_0 \dfrac{\partial I^2}{\partial t^2} - \dfrac{\partial }{\partial x} \left( c \dfrac{\partial I}{\partial x} \right) - \dfrac{\partial \dot{D}}{\partial x} =0
\end{equation}
Since $\dot{D}$ is functionally independent of $I$, equation \eqref{6.4} is the Euler-Lagrange equations associated with the functional   
\begin{equation} \label{6.5}
\Sigma(I) = \int_{\mathcal{B}} \sigma (\partial_t I, \partial_x I,x,t) \, dx \doteq \int_{\mathcal{B}} \left(\frac{ \rho_0}{2} \left( \dfrac{\partial I}{\partial t} \right)^2 - \frac{c}{2} \left( \dfrac{\partial I}{\partial x} \right)^2 - \dot{D} \dfrac{\partial I}{\partial x} \right) \, dx
\end{equation}
\end{flushleft}

\begin{flushleft}
\justify 
Therefore, the evolution of the total energy $I$ in an arbitrary interval $[0, \tau]$ coincides with the stationary points of 
\begin{equation} \label{6.6}
\int_{0}^{\tau} \Sigma dt
\end{equation}
provided the boundary conditions $\delta I \mid _{t = 0} = \delta I \mid_{\tau =0} = 0$ are satisfied. This is the least action principle in the dissipative case.
\end{flushleft}

\begin{remark}
\justify
\leavevmode
\begin{enumerate} [i]
\item The quantity $\dot{D}$ should be thought of as a non-homogeneous term independent of $I$. This is not unusual in the calculus of variations; for example consider the following Dirichlet functional:

$$ \Phi (w) = \int_{\mathcal{B}} \phi(\nabla w, w, x) dx = \int_{\mathcal{B}} \dfrac{1}{2} |\nabla w|^2 - wf(x) dx  $$

So while $\phi$ includes the non-homogeneous term $f$, it is functionally dependent on only the basic fields $\nabla \phi$ and $\phi$. In our case in \eqref{6.5}, our basic fields include the derivatives of $I$ alone. Therefore, the term $\dot{D}$ is comparable to $f$ in the above Dirichlet functional and is independent of $I$. Another example is given in Table 1 below. 

\item The Euler-Lagrange equations for \eqref{6.5} share the same form with the conservative case since the variation is taken only with respect to the total energy I
\begin{equation} \label{6.7}
\dfrac{\partial}{\partial t} \left( \dfrac{\delta \Sigma }{\delta (\partial_t I) } \right) - \dfrac{\delta \Sigma}{\delta I} = 0 \ .
\end{equation}
Furthermore, the evolution of the dissipative material is entirely governed by a set of two coupled Euler-Lagrange equations: \eqref{6.7} coupled with 
\begin{equation}  \label{6.7a}
\dfrac{\partial }{\partial t}  \left( \dfrac{\delta \mathcal{L}}{\delta (\partial_t u)} \right) - \dfrac{\delta \mathcal{L}}{\delta u} = 0 \ ,
\end{equation}
where the Lagrangian $\mathcal{L}$ is
$$ \mathcal{L} =  \int_{\mathcal{B}} \dfrac{\rho_0}{2}  \left(\dfrac{\partial u}{\partial t} \right)^2 - e(\partial_x u ,s) d x \ .$$
Therefore, the functionals $\Sigma$ and $\mathcal{L}$ completely determine the state of the dissipative system at each instant in time, and  each must be varied with respect to $u$ and $I$, respectively, to obtain the evolution equations. As is standard, in order to obtain the unique physically correct solution (at least for short time), the coupled system \eqref{6.7}--\eqref{6.7a} needs to be supplemented by appropriate initial-boundary conditions.
\end{enumerate}
\end{remark}

\begin{flushleft}
\justify
In the example below, we obtain a special form of functional $\Sigma$. Throughout the example we assume that the reference configuration is stress free i.e. $S(\partial_x u, s)\mid_{0 } =0$, and the temperature at the reference configuration $\theta \mid_{0 } = \theta_0$, where we have written $\psi \mid_{0 } $ to mean that the function $\psi$ is evaluated at $\partial_x u = 0$ and $s = 0$.    \\
\end{flushleft}

\emph{Example (heat flow in a thermoelastic medium)}

\begin{flushleft}
\justify
We begin first by writing $\sigma$ as the sum of two quantities:  $\sigma = \sigma_h + R_h$, where we have defined

\begin{equation*} \label{6.8}
\sigma_h =  \dfrac{ \rho_0}{2}\left( \dfrac{\partial I}{\partial t} \right)^2 + \dfrac{\partial q}{\partial t}\dfrac{\partial I}{\partial x}
\end{equation*}

\begin{equation*} \label{6.9}
R_h =  - \frac{c}{2} \left( \dfrac{\partial I}{\partial x} \right)^2 + \left( \dfrac{\partial S}{\partial s} v \dfrac{\partial s}{\partial t} + \dfrac{\partial s}{\partial x} S^2 \dfrac{\partial }{\partial (\partial_x u)} \left( \dfrac{\theta}{S} \right) \right) \dfrac{\partial I}{\partial x}
\end{equation*}
\end{flushleft}

\begin{flushleft}
\justify
We examine the \emph{linear} theory of thermoelasticity which is charechtrized by ``small'' thermomechanical deformations. Thereofore, it is appropriate to rescale the displacement field $\epsilon u$ for suitable non-dimensional positive small parameter $\epsilon $ \cite{dal2002linearized}, and consider the temperature $\theta$ to be everywhere close to $\theta_0$. Since both the displacement field and the temperature undergo small changes, the entropy would also be rescaled to $\epsilon s$. Linearizing the theory entails keeping terms up to the order of $\epsilon^2$ in the functional $\Sigma$, while terms up to the order of $\epsilon$ alone are considered in the Euler-Lagrange equations. 
\end{flushleft}

\begin{flushleft}
\justify
By expanding the constitutive equations for the stress and temperature around $\partial_x u =0$ and $s=0$, while recalling $S \mid_0  =0$, we obtain
\begin{equation*}   \label{6.10}
S(\partial_x u,s; \epsilon) =  \dfrac{\partial S}{\partial (\partial_x u)}  \mid_0 \epsilon \dfrac{\partial u}{\partial x} + \dfrac{\partial S}{\partial s}  \mid_0 \epsilon s + O(\epsilon^2) 
\end{equation*}
\begin{equation}  \label{6.11}
\theta(\partial_x u,s; \epsilon) =\theta_0 +\dfrac{\partial \theta}{\partial (\partial_x u)}  \mid_0  \epsilon \dfrac{\partial u}{\partial x} +  \dfrac{\partial \theta}{\partial s}  \mid_0 \epsilon s + O(\epsilon^2) 
\end{equation}
\end{flushleft}

\begin{flushleft}
\justify
By substituting the above linearizations together with the rescaled displacement and entropy fields into $R_h$, the  second order approximation (for constant $c$) gives us:
\begin{equation} \label{6.12}
R_h = - \dfrac{c}{2}\left( \dfrac{\partial I}{\partial x} \right)^2 + \theta_0 c \dfrac{\partial s}{\partial x} \dfrac{\partial I}{\partial x}
\end{equation}
While we have not rescaled the total energy explicitly, a simple calculation reveals that up to the order of $\epsilon$, we have $\dfrac{\partial I}{\partial x} = \theta_0 \dfrac{\partial s}{\partial x}$. Therefore, it is justified to retain the terms involving the total energy in \eqref{6.12}. However, for the stationary principle to hold (see equation \eqref{6.6}), we maintain the basic form of $R_h$ in terms of total energy field $I$ and not the displacement and entropy fields.  
\end{flushleft}

\begin{flushleft}
\justify
The Euler-Lagrange equations for the system are
\begin{align*}
\dfrac{\partial}{\partial t} \left( \dfrac{\delta (\Sigma_h +\mathcal{R}_h) }{\delta \dot{I}} \right) - \dfrac{\delta (\Sigma_h + \mathcal{R}_h)}{\delta I} = 0 \\
\end{align*}
which reduce to

\begin{equation*} \label{6.13}
\dfrac{\partial}{\partial t} \left( \dfrac{\delta \Sigma_h }{\delta \dot{I}} \right) - \dfrac{\delta (\Sigma_h + \mathcal{R}_h)}{\delta I} = 0 
\end{equation*}
However, in the linear approximation (i.e. up to the order of $\epsilon$) we have in fact
\begin{equation*} \label{6.14}
\dfrac{\delta  \mathcal{R}_h}{\delta I} = 0
\end{equation*}
Therefore, in the linear approximation we get
\begin{equation*} \label{6.15}
\dfrac{\partial}{\partial t} \left( \dfrac{\delta \Sigma_h }{\delta \dot{I}} \right) - \dfrac{\delta \Sigma_h}{\delta I} = 0
\end{equation*}
\end{flushleft}

\begin{flushleft}
\justify
Functional $\Sigma_h$ determines the evolution of the energy for a dissipative linear thermoelastic medium. Equation \eqref{6.15} reads
\begin{align*} \label{6.16}
\dfrac{\partial }{\partial t} \left( \rho_0 \dfrac{\partial I}{\partial t} +  \dfrac{\partial q}{\partial x} \right) = \dfrac{\partial }{\partial t} \left( \rho_0 \theta_0 \dfrac{\partial s}{\partial t} + \dfrac{\partial q}{\partial x} \right)  = 0
\end{align*}
\end{flushleft}

\begin{flushleft}
\justify
We have obtained the classic entropy balance as a constant of motion. Substituting for the linear constitutive law \eqref{6.11} (up to the order of $\epsilon$) together with Fourier's law $q = -k\nabla \theta$, we produce the classic evolution-diffusion equation of thermoelasticity
\begin{equation} \label{6.17}
\dfrac{\partial }{\partial t}\left( \rho_0 c_0 \dfrac{ \partial \theta}{\partial t} - \rho_0 \gamma \theta_0  \dfrac{\partial^2 u}{\partial x \partial t} - k \nabla^2 \theta \right)  = 0
\end{equation}
where, $c_0 = \frac{\theta_0}{\partial \theta / \partial s} \mid_0 $, $\gamma =\frac{c_0}{\theta_0} \frac{\partial \theta}{\partial (\partial_x u)} \mid_0 $, and $k$ are the heat capacity, stress-temperature modulus, and conductivity constant, respectively. If we assume the fields $\theta$ and $u$ have compact support on $\mathcal{B}$ then the classic evolution-diffusion equation of thermoelasticity can be readily recovered. 

\end{flushleft}

\begin{flushleft}
\justify
In the absence of mechanical processes, that is $u = 0$, equation \eqref{6.17} reduces to the heat equation. In fact, in this case $c = 0$, $\dot{D}= - \partial_t q $, which gives $R_h = 0$. The heat equation is, therefore, also obtained from the functional $\Sigma_h$. This concludes our example.\\
\end{flushleft}

\begin{remark}
\justify
\leavevmode
Since we have used Fourier's law to model heat conduction, it is expected that the classical equations of thermoelasticity and the heat equation follow; however, this is not the only possible model. For example, in the framework of Rational Extended Thermodynamics the quantity $q + k\, \nabla \theta \neq 0 $, but is proportional to $\partial_t q$, the proportionality constant is called the \emph{relaxation time}\footnote{This is sometime referred to as Cattaneo’s Law.}. Within this framework the heat equation has an additional hyperbolic term which, for long time scales, describes a reversible process \cite{jou1996extended}. 

\end{remark}

\begin{flushleft}
\justify
The Hamiltonian and bracket formalisms for dissipative systems follow similarly as in Section 3. Define the Legendre transform of $\sigma$: 
\begin{equation} \label{6.18}
\pi(\nabla I,J) \doteq J \cdot \dot{I} - \sigma
\end{equation}
which gives 
\begin{equation*} \label{6.19}
\Pi = \int_{\mathcal{B}} \left( \dfrac{1}{2\rho_0} J^2 +\dfrac{c}{2} \left( \dfrac{\partial I}{\partial x} \right)^2 + \dot{D} \dfrac{\partial I}{\partial x}  \right) d x
\end{equation*} 
and the evolution equations can be rewritten once again in terms of $\Pi$:

\begin{align*}
\dot{I}= \dfrac{\delta \Pi}{\delta J} \\
\dot{J}  = - \dfrac{\delta \Pi}{\delta I}
\end{align*}

or in terms of the Poisson bracket:

\begin{align*}
\dot{I}= \lbrace I, \Pi \rbrace  \\
\dot{J}  = \lbrace J, \Pi \rbrace
\end{align*}

where the $\lbrace,\rbrace$ is defined \eqref{3.6} and the phase space is given in Section 3. 
\end{flushleft}

\begin{flushleft}
\justify
Lastly, we consider some variational symmetries. If we consider the one-parameter transformation in terms of the total energy: $ \lambda \longmapsto I_{\lambda} = I(z_{\mu}; \lambda)$ with $\sigma_\lambda = \sigma (\partial_{\mu} I_\lambda)$, then it is not hard to see that Noether's theorem, as presented in Section 4, holds. Hence, the simple transformation: $I_\lambda = I+\lambda\overline{I}$ leaves $\sigma$ invariant and so, by equation \eqref{4.3}--\eqref{4.4}, we obtain
$$\dfrac{\partial }{\partial z_{\mu} } \left(  \dfrac{\partial \sigma }{\partial (\partial_{\mu} I)} \right) = 0 $$ 
which is nothing other than the energy balance in hyperbolic form \eqref{6.4}.
\end{flushleft}

\begin{flushleft}
\justify
We examine, once again, the basic space-time translations (for simplicity we consider $c$ constant--see Appendix A for the non-constant $c$)
\begin{equation} \label{6.22}
\lambda \longmapsto z^\lambda_\mu = z_\mu - \lambda_\mu = z_\mu - \lambda_\eta \delta_{\mu \eta}
\end{equation}
The question here is whether or not $\sigma$ remains invariant under the action of this transformation now that $\sigma$ depends not only on the parameterized field $I_\lambda$ but also on $\dot{D}_\lambda$. However, for the conservation law \eqref{4.4} to hold, the calculation in \eqref{4.5} must be justified. As shown below this is not true  for the basic space-time translations \eqref{6.22}. This fact will result, as is shown below, in non-homogeneous conservation laws, which are still valuable in the context of constructing a variational hierarchy (see Appendix A).
\end{flushleft}

\begin{flushleft}
\justify
Under space-time translations:

\begin{subequations} \label{6.23}
\begin{align} 
I_\lambda &= I(z_\mu - \lambda_\mu) = I(z_\mu) + \lambda_\mu \dfrac{\partial I_{\lambda}}{\partial z_\mu} \mid_{\lambda = 0} + O(| \lambda_\mu|^2) \label{6.23a} \\
\dot{D}_\lambda &= \dot{D} (z_\mu - \lambda_\mu) = \dot{D}(z_\mu) + \lambda_\mu \dfrac{\partial \dot{D}_{\lambda}}{\partial z_\mu} \mid_{\lambda = 0} + O(| \lambda_\mu|^2) \label{6.23b}
\end{align}
\end{subequations}
\end{flushleft}

\begin{flushleft}
\justify
Next, we show $\sigma_{\lambda}$ satisfies \eqref{4.2}
\begin{proposition}  \label{prop2}
Function $\sigma$ is invariant under the space-time translation \eqref{4.2}.
\end{proposition} 
\end{flushleft}

\begin{proof}
Function $\sigma_{\lambda}$ is defined as $\sigma(z^\lambda_\mu)$. Expanding this expression in $\lambda_\mu$ yields
\begin{align*}
 \sigma_{\lambda}  &= \sigma(z_\mu) + \lambda_\mu \dfrac{d}{d \lambda_\mu} \sigma_{\lambda} \mid_{\lambda = 0} +  O(| \lambda_\mu|^2)\\
&= \sigma + \lambda_\mu \left( \dfrac{\partial \sigma}{\partial (\partial_\alpha I)} \dfrac{\partial (\partial_\alpha I)}{\partial z_\mu }  + \dfrac{\partial \sigma}{\partial \dot{D} } \dfrac{\partial\dot{D}}{\partial z_\mu }  \right) +  O(| \lambda_\mu|^2)  \\
&= \sigma + \lambda_\mu \dfrac{\partial \sigma}{\partial z_\mu} +O(| \lambda_\mu|^2)
\end{align*}

Hence, 
\begin{equation} \label{6.25}
\dfrac{d}{d \lambda_{\eta}} \sigma_\lambda \mid_{\lambda = 0} = \dfrac{\partial }{\partial z_{\mu}} \left( \delta_{\mu \eta} \sigma \right)
\end{equation}
Here again since we have assumed a diffeomorphism with respect to a two dimensional parameter $\lambda_\eta$, the vector field in \eqref{4.2} is augmented to a second order tensor. Therefore, according to \eqref{6.25}, we take $K_{\mu \eta} = \delta_{\mu \eta} \sigma$ to satisfy the condition \eqref{4.2}.

\end{proof}

\begin{flushleft}
\justify
However, we can not apply Noether's theorem to conclude the conservation laws associated with the energy-momentum tensor since the LHS of \eqref{6.25} does not result in \eqref{4.5} in the presence of $\dot{D}_\lambda$. Nevertheless, we can establish a non-homogeneous version of equation \eqref{4.3}, that is with additional terms on the RHS. 
\end{flushleft}

\begin{flushleft}
\justify
We start out calculation with
\begin{align*}
\dfrac{d}{d \lambda_{\eta}} \sigma_\lambda \mid_{\lambda = 0} &= \dfrac{\partial \sigma}{\partial (\partial_\mu I_\lambda)} \dfrac{\partial (\partial_\mu I_\lambda)}{\partial \lambda_\eta} \mid_{\lambda = 0} + \dfrac{\partial \sigma}{\partial \dot{D}_\lambda} \dfrac{\partial \dot{D}_\lambda}{\partial \lambda_\eta} \mid_{\lambda = 0} \\
&= - \dfrac{\partial }{\partial z_{\mu} } \left(  \dfrac{\partial \sigma }{\partial (\partial_{\mu} I)} \right)  \dfrac{\partial I_\lambda}{\partial \lambda_\eta } \mid_{\lambda=0} + \dfrac{\partial  } {\partial z_{\mu}} \left(  \dfrac{\partial \sigma }{\partial (\partial_{\mu} I)}  \dfrac{\partial I_\lambda}{\partial \lambda_\eta } \mid_{\lambda=0}  \right) + \dfrac{\partial \sigma}{\partial \dot{D}} \dfrac{\partial \dot{D}_\lambda}{\partial \lambda_\eta} \mid_{\lambda = 0} \\
&= \dfrac{\partial  } {\partial z_{\mu}} \left(  \dfrac{\partial \sigma }{\partial (\partial_{\mu} I)}  \dfrac{\partial I}{\partial z_\eta}  \right) + \dfrac{\partial \sigma}{\partial \dot{D}} \dfrac{\partial \dot{D}}{\partial z_\eta}  \numberthis \label{6.26}
\end{align*}
\end{flushleft}

\begin{flushleft}
\justify
Equations \eqref{6.25} and \eqref{6.26} give us the following equation for the energy momentum tensor $T_{\mu \eta} $ (defined in \eqref{4.15})
\begin{equation*} \label{6.27}
\dfrac{\partial T_{\mu \eta}}{\partial z_\mu} = - \dfrac{\partial \sigma}{\partial \dot{D}} \dfrac{\partial \dot{D}}{\partial z_\eta}
\end{equation*}
In particular, the quantity $T_{00} = \pi $ is governed by
\begin{equation} \label{6.28}
\dfrac{\partial \pi}{\partial t} + \dfrac{\partial }{\partial x} \left(  \dfrac{\partial I}{\partial t} \dfrac{\partial \sigma}{\partial (\partial_x I)}  \right) =  \dfrac{\partial I}{\partial x}\dfrac{\partial \dot{D}}{\partial t}
\end{equation}
It is also clear that the tensor $T_{\mu \eta}$ is no longer symmetric. 
\end{flushleft}

\begin{flushleft}
\justify
The quantity $\pi$ in the dissipative case is no longer governed by a conservation law as evident by equation \eqref{6.28}. Nevertheless, a similar procedure outlined in the calculations \eqref{6.1} and \eqref{6.2} applied to the flux $\frac{\partial I}{\partial t} \frac{\partial \sigma}{\partial (\partial_x I)}$ can produce the functional $\mathcal{L}^2$ for the dissipative case. So while we no longer have conservation laws given by Noether's theorem in the presence of dissipation, we still manage to retain an infinite number of variational principles together with their corresponding Lagrangian-Hamiltonian formalism associated with functional $\mathcal{L}^i$ (see Remark 4).
\end{flushleft}

\begin{flushleft}
\justify
We summarize this section by demonstrating the structural symmetry that exists between the conversation of momentum and conversation of energy in continuum mechanics in the presence of thermal effects. 
\end{flushleft}

\begin{table}[h!]
\caption{A one-to-one correspondence is illustrated between the balance of momentum and energy in linear elastic medium with the presence of thermal effects. For simplicity we have taken stress-temperature modulus to be equal to one. }
\hspace{+12mm} {\renewcommand{\arraystretch}{4.5} 
\begin{tabular}{|c|c|c|}
\hline
\textbf{Identity} & \textbf{Conservation of Momentum} & \pbox{20cm}{\textbf{Conservation of Energy}}\\\hline
Density Function &  $L = \dfrac{ \rho_0}{2} \left( \dfrac{ \partial u}{\partial t} \right)^2 - \dfrac{c}{2}  \left( \dfrac{\partial u}{\partial x} \right)^2  - \theta\dfrac{\partial u}{\partial x} $ & $\sigma =  \dfrac{ \rho_0}{2} \left( \dfrac{ \partial I}{\partial t} \right)^2  - \dfrac{c}{2} \left( \dfrac{\partial I}{\partial x} \right)^2  - \dot{D}  \dfrac{\partial I}{\partial x}  $ \\\hline
Governing Equation &  $\rho_0 \dfrac{\partial^2 u}{\partial t}  - c \dfrac{\partial^2 u} {\partial x^2}  - \dfrac{\partial \theta}{ \partial x}  = 0 $ & $\rho_0 \dfrac{\partial^2 I}{\partial t^2}   - c \dfrac{\partial^2 I} {\partial x^2}  - \dfrac{\partial \dot{D}}{ \partial x}  = 0 $  \\\hline
Action Functional & $ \int_{0}^{\tau} \int_{\mathcal{B}} L(\partial_t u, \partial_x u,x,t) dx dt$ &  $ \int_{0}^{\tau} \int_{\mathcal{B}} \sigma(\partial_t I, \partial_x I,x,t ) dx dt$ \\\hline

\end{tabular}
}
\end{table}

\begin{flushleft}
\justify
The above table establishes the following correspondence 

\begin{align*} 
u  &\longleftrightarrow I \\
\theta  &\longleftrightarrow \dot{D} 
\end{align*}

\end{flushleft}

\section{Discussion}

\begin{flushleft}
\justify

The results we have obtained in this work fall under the rubric of Rational Thermodynamics where the dissipative thermoelastic material is defined by the constitutive laws \eqref{2.2}. Within this framework we succeeded in advancing a complete analytical mechanics for the field $I$, which followed once the variational structure of the first law of thermodynamics was demonstrated. In the section we wish to highlight some main difference between our work and variational principles that are established in classical irreversible thermodynamics. Namely we look at the variational principle obtained by Onsager, which produces the well-known reciprocal relations and the general linear constitutive relations between the thermodynamic forces and fluxes.

\end{flushleft}

\begin{flushleft}
\justify

The work of Lars Onsager \cite{onsager1931reciprocal} contained one of the earliest attempts of applying variational principles to thermodynamic systems outside of equilibrium. The validity of Onsager's ``Least Dissipation of Energy Principle"  depends on the construction of a dissipation function $\Psi (\vec{X}, \vec{X})$: a quadratic form in the thermodynamic forces. Then according to Onsager, the true values of the forces $\vec{X}$, given the fluxes $\vec{J}$, must satisfy
\begin{equation} \label{7.1}
\dot{\eta} (\vec{X}) - \Psi(\vec{X}, \vec{X}) = \mathrm{max}
\end{equation}
where $\dot{\eta} (X) = \sum_i X_i J_i$ is the internal entropy production density of the system, and $\Psi (\vec{X}, \vec{X})= \sum_{i,k} \frac{1}{2} R_{ik} X_i X_k $. The stationary points of the functional \eqref{7.1} give the well-known linear kinematical constitutive relations between the fluxes $\vec{J}$ and forces $\vec{X}$, and Onsager's reciprocity condition, respectively:
\begin{equation} \label{7.2}
 J_i = \sum_k R_{ik} X_k \, \, \, \,  \mathrm{and} \, \, \, \, \, R_{ik}=R_{ki} 
\end{equation}
The Euler-Lagrange equations, therefore, do not produce the equations of motion, they determine the admissible forces (or fluxes). However, for simple systems, such as heat conduction in a solid body, it is possible to derive the temperature distribution for the \emph{stationary states} only by inserting the linear constitutive relations \eqref{7.2}$_1$ into the variational principle \eqref{7.1} \cite{gyarmati1970non}. More explicitly, consider Fourier's law for modeling heat conduction
\begin{equation*} 
\vec{q} = -k\nabla \theta
\end{equation*}
or equivalent for non-zero $\theta$
\begin{equation*} 
\vec{q} = k \,  \theta^2  \,\nabla \left( \dfrac{1}{\theta} \right)
\end{equation*}
Equation \eqref{7.2}$_1$ is satisfied with $\vec{J} = \vec{q}$, $R_{ik} = k \, \theta^2 \delta_{ik}$ and $\vec{X} = \nabla \left( \dfrac{1}{\theta} \right)$. With these identifications, \eqref{7.1} applied to a continuum $\mathcal{B}$ gives

\begin{equation*}
\delta \int_{\mathcal{B}} \dfrac{R_{kk}}{2} \left ( \nabla \dfrac{1}{\theta} \right)^2 dx = 0
\end{equation*}
Taking the variation with respect to $\theta$ and assuming variations of compact support over $\mathcal{B}$, the Euler-Lagrange equations produce
$$ \Delta \left( \dfrac{1}{\theta} \right) =0 $$
which is the stationary heat equation.
\end{flushleft}

\begin{flushleft}
\justify

Our variational principle on the other hand, as formulated in Section 6, gives the \emph{evolution} of the dissipative system and is not confined to the stationary case. However, the underlying conception behind these two approaches is very different: the least dissipation of energy principle can be viewed as an extension of the second law of thermodynamic which selects the correct stationary thermodynamic state, whereas our least action principle is based on a variational structure already inherent in the equations of motion which follows once the balance of momentum is established. As Remark (1-ii) and the example in Section 6 illustrate, the dynamical equations can be derived variationally when the variational principle is formulated in terms of the field $I$ and not the temperature $\theta$. 

\end{flushleft}

\section{Conclusion}

\begin{flushleft}
\justify
In this work we have constructed a least action principle in 1D for the nonlinear conservation of energy equation  in the absence and presence of entropy production due to heat transfer. An action functional was constructed by revealing that the rate of change in the energy flux can be written in terms of the sum of the gradient of the total energy and a dissipative term $\dot{D}$. This natural method enabled us to recast the energy equation into second order hyperbolic form. This new action $\int_{0}^{\tau} \Sigma dt$ is akin to Hamilton's principle for least action and it determines the evolution of the energy as the functional's extremals. A slew of consequences followed; namely a complete analytic mechanics for the balance of energy equation in terms of the total energy $I$.  We also showed that infinite new variational principles follow from the classic Hamilton's principle via Noether's theorem. This hierarchical structure has a simple iterative form for isentropic (conservative) systems, and is more involved for dissipative and forced systems. The formalism developed in this work was applied to the classical thermoelastic model defined by the constitutive laws \eqref{2.2}. The application of the formalism to higher dimensions and to other cases such rheological fluids and viscoelastic/plastic materials will be topics of future study.
\end{flushleft}

\begin{flushleft}

\end{flushleft}

\begin{flushleft}
\justify
\large \textbf{Acknowledgments.} \normalsize  The author thanks James Glimm for helpful discussions and encouragement. This work was supported and funded by Kuwait University Research Grant No. [ZS02/19].
\end{flushleft}
\clearpage

\appendix
\section{Appendix }

\begin{flushleft}
\justify
In this paper we have asserted that the scheme for constructing the functional $\Sigma$ can accomedate external sources and the procedure of which is similiar to that of Section 6. We, therefore, consider the most general case of the first law of thermodyanics applied to a continuum, that is the non-linear conservation of energy equation with body forces $b(x,t)$, heat sources $r(x,t)$, and heat flux $q = - k \frac{\partial \theta}{\partial x}$ \cite{dafermos2005hyperbolic}
\begin{equation} \label{A1}
\rho_0 \dfrac{\partial I}{\partial t} = \dfrac{\partial }{\partial x} \left( vS - q \right) +\rho_0 vb + \rho_0 r \qquad \text{in} \; \mathcal{B} \times [0,\tau]
\end{equation}
where the constitutive equations are given in \eqref{2.2}. Moreover, we assume through-out a non-constant elastic modulus $c = \dfrac{\partial S}{\partial (\partial_x u)}$.
\end{flushleft}

\begin{flushleft}
\justify
In addition to the construction of $L^1 = \Sigma$ for the above problem (together with its corresponding Hamiltonian formulation and Noether's theorem), we also demonstrate the construction of $L^2$ in this context.
\end{flushleft}

\begin{flushleft}
\justify
By the calculations in \eqref{6.1} and \eqref{6.2} we have
\begin{equation*} \label{A2}
\dfrac{\partial }{\partial t} \left( vS - q \right) = c \dfrac{\partial I}{\partial x} + \dot{D}
\end{equation*}
The above two equations thus give
\begin{equation}  \label{A3}
\rho_0 \dfrac{\partial^2 I}{\partial t^2} = \dfrac{\partial }{\partial x} \left( c \dfrac{\partial I}{\partial x} \right) + \dfrac{\partial \dot{D}}{\partial x} - \dot{B}
\end{equation}
where we have defined $ \dot{B} =- \rho_0 \frac{\partial }{\partial t} \left( vb - r \right)$. 
\end{flushleft}

\begin{flushleft}
\justify
As such we choose functional $\Sigma$ to be
\begin{equation*} \label{A4}
\Sigma = \int_{\mathcal{B}} \left(\frac{ \rho_0}{2} \left( \dfrac{\partial I}{\partial t} \right)^2 - \frac{c}{2} \left( \dfrac{\partial I}{\partial x} \right)^2 - \dot{D} \dfrac{\partial I}{\partial x}  - \dot{B} I \right) \, dx
\end{equation*}
Hence, equation \eqref{A3} is equivalent to the standard Euler-Lagrange equations
\begin{equation} \label{A5}
\dfrac{d }{d t} \left( \dfrac{\delta \Sigma}{\delta (\partial_t I)} \right) - \dfrac{\delta \Sigma}{\delta I } =0
\end{equation}
\end{flushleft}

\begin{flushleft}
\justify
The Hamiltonian density $\pi$ can be fashioned similarly to Sections 3 \& 6. Indeed, we can combine $\dot{B}$ and $\partial_x \dot{D}$ into one term and simply apply the procedure in Section 6 to obtain
\begin{equation*} \label{A6}
\Pi = \int_{\mathcal{B}} \left( \dfrac{1}{2\rho_0} J^2 +\dfrac{c}{2} \left( \dfrac{\partial I}{\partial x} \right)^2 + \dot{D} \dfrac{\partial I}{\partial x}  + \dot{B} I \right) d x
\end{equation*}
and

\begin{align*}
\dfrac{\partial I}{\partial t} = \dfrac{\delta \Pi}{\delta J}  \\
\dfrac{\partial J}{\partial t}   = - \dfrac{\delta \Pi}{\delta I}
\end{align*}
where $\pi$ is defined as in equation \eqref{6.18}.
\end{flushleft}

\begin{flushleft}
\justify
Equation \eqref{A3} (or equivalently equation \eqref{A5}) follows from Noether's theorem if $\sigma$ remains invariant with respect to $\lambda \longrightarrow I_\lambda = I +\lambda \overline{I}$. However, the application of Noether's theorem is not straightforward with respect to space-time translations \eqref{4.10} because of the presence of $\dot{D}$ (as we have seen in Section 6), and now $\dot{B}$. We explore this below.
\end{flushleft}

\begin{flushleft}
\justify
Under space-time translations \eqref{4.10} we obtain two more equations in addition to \eqref{6.23}

\begin{align*} 
c_\lambda &= c(z_\mu - \lambda_\mu) = c(z_\mu) + \lambda_\mu \dfrac{\partial c_{\lambda}}{\partial z_\mu} \mid_{\lambda = 0} + O(| \lambda_\mu|^2)  \\
B_\lambda &= B (z_\mu - \lambda_\mu) = B(z_\mu) + \lambda_\mu \dfrac{\partial B_{\lambda}}{\partial z_\mu} \mid_{\lambda = 0} + O(| \lambda_\mu|^2)
\end{align*}

\end{flushleft}

\begin{flushleft}
\justify
A similar calculation to the proof Proposition \eqref{prop2} reveals that $\sigma_\lambda$ obeys
\begin{equation*} \label{A9}
\dfrac{d}{d \lambda_{\eta}} \sigma_\lambda \mid_{\lambda = 0} = \dfrac{\partial }{\partial z_{\mu}} \left( \delta_{\mu \eta} \sigma \right)
\end{equation*}
\end{flushleft}

\begin{flushleft}
\justify
On the other hand 

\begin{equation*} 
\begin{aligned}
 \dfrac{d}{d \lambda_{\eta}} \sigma_\lambda \mid_{\lambda = 0} &= \dfrac{\partial \sigma}{\partial I_\lambda } \dfrac{\partial  I_\lambda}{\partial \lambda_\eta} \mid_{\lambda = 0} + \dfrac{\partial \sigma}{\partial (\partial_\mu I_\lambda)} \dfrac{\partial (\partial_\mu I_\lambda)}{\partial \lambda_\eta} \mid_{\lambda = 0} + \dfrac{\partial \sigma}{\partial c_\lambda} \dfrac{\partial c_\lambda}{\partial \lambda_\eta} \mid_{\lambda = 0} + \dfrac{\partial \sigma}{\partial \dot{D}_\lambda} \dfrac{\partial \dot{D}_\lambda}{\partial \lambda_\eta} \mid_{\lambda = 0} + \dfrac{\partial \sigma}{\partial \dot{B}_\lambda} \dfrac{\partial \dot{B}_\lambda}{\partial \lambda_\eta} \mid_{\lambda = 0} \\
 &= \left( \dfrac{\partial \sigma}{\partial I } - \dfrac{\partial }{\partial z_{\mu} } \left(  \dfrac{\partial \sigma }{\partial (\partial_{\mu} I)} \right) \right)  \dfrac{\partial I_\lambda}{\partial \lambda_\eta } \mid_{\lambda=0} + \dfrac{\partial  } {\partial z_{\mu}} \left(  \dfrac{\partial \sigma }{\partial (\partial_{\mu} I)}  \dfrac{\partial I_\lambda}{\partial \lambda_\eta } \mid_{\lambda=0}  \right)  + \cdots \\
 &= \dfrac{\partial  } {\partial z_{\mu}} \left(  \dfrac{\partial \sigma }{\partial (\partial_{\mu} I)}  \dfrac{\partial I_\lambda}{\partial \lambda_\eta } \mid_{\lambda=0}  \right) + \dfrac{\partial \sigma}{\partial c} \dfrac{\partial c_\lambda}{\partial \lambda_\eta} \mid_{\lambda = 0} + \dfrac{\partial \sigma}{\partial \dot{D}} \dfrac{\partial \dot{D}_\lambda}{\partial \lambda_\eta} \mid_{\lambda = 0} + \dfrac{\partial \sigma}{\partial \dot{B}} \dfrac{\partial \dot{B}_\lambda}{\partial \lambda_\eta} \mid_{\lambda = 0} \\
 &= \dfrac{\partial  } {\partial z_{\mu}} \left(  \dfrac{\partial \sigma }{\partial (\partial_{\mu} I)}  \dfrac{\partial I}{\partial z_\eta}  \right)+ \dfrac{\partial \sigma}{\partial c} \dfrac{\partial c}{\partial z_\eta}  + \dfrac{\partial \sigma}{\partial \dot{D}} \dfrac{\partial \dot{D}}{\partial z_\eta} +  \dfrac{\partial \sigma}{\partial \dot{B}} \dfrac{\partial \dot{B}}{\partial z_\eta} 
\end{aligned}
\end{equation*}

\end{flushleft}

\begin{flushleft}
\justify
Therefore, the second order tensor $T_{\mu \eta} = \frac{\partial \sigma}{\partial (\partial_\mu I)} \frac{\partial I}{\partial z_\eta} - \sigma \delta_{\mu \eta} $ satisfies the equation
\begin{equation*} \label{A11}
\dfrac{\partial T_{\mu \eta}}{\partial z_\mu } -  F_\eta =0
\end{equation*}
where $F_\eta =  \dfrac{\partial \sigma}{\partial c} \dfrac{\partial c}{\partial z_\eta}  + \dfrac{\partial \sigma}{\partial \dot{D}} \dfrac{\partial \dot{D}}{\partial z_\eta} +  \dfrac{\partial \sigma}{\partial \dot{B}} \dfrac{\partial \dot{B}}{\partial z_\eta} $.
\end{flushleft}

\begin{flushleft}
\justify
And so $T_{00} = \pi$ satisfies
\begin{equation} \label{A12}
\dfrac{\partial \pi}{\partial t} + \dfrac{\partial }{\partial x} \left( \dfrac{\partial I}{\partial t} \dfrac{\partial \sigma}{\partial (\partial_x I)} \right) = F_0
\end{equation}
\end{flushleft}

\begin{flushleft}
\justify
The presence of $F_0$ does not impede the construction of $L^2$ (i.e. the second iteration). We proceed as we did in the beginning of Section 5, and compute the rate of change of the flux in \eqref{A12}
\begin{align*}
\dfrac{\partial }{\partial t} \left( \dfrac{\partial I}{\partial t} \dfrac{\partial \sigma}{\partial (\partial_x I)} \right) &= - \dfrac{\partial }{\partial t} \left( \dfrac{\partial I}{\partial t} \left( c\dfrac{\partial I}{\partial x} + \dot{D} \right)  \right) \\
&= - \dfrac{\partial^2 I}{\partial t^2} \left( \dot{D} + c\dfrac{\partial I}{\partial x} \right) - \dfrac{\partial I}{\partial t} \left( \dfrac{\partial c}{\partial t} \dfrac{\partial I}{\partial x} + c \dfrac{\partial^2I}{\partial x \partial t} + \dfrac{\partial \dot{D}}{\partial t} \right)
\end{align*}
By using the evolution equation for the total energy \eqref{A3} and rearranging the terms we obtain
\begin{align} \label{A13}
\dfrac{\partial }{\partial t} \left( \dfrac{\partial I}{\partial t} \dfrac{\partial \sigma}{\partial (\partial_x I)} \right) = -\dfrac{c}{\rho_0} \dfrac{\partial }{\partial x} \left( \pi \right) + Q
\end{align}
where $Q = \frac{-1}{\rho_0} \partial_ x \left( c \partial_x I \right) \dot{D} - \frac{c}{2\rho_0} \partial_x c (\partial_x I)^2 \frac{\dot{D}}{\rho_0}\partial_x \dot{D} + \frac{c}{\rho_0} \partial_{xx} I \dot{D}+\frac{\dot{D}\dot{B}}{\rho_0} - \frac{c}{\rho_0} I \partial_x\dot{B} + \frac{2c}{\rho_0} \partial_x(I \dot{B}) - \partial_t I \partial_t c \partial_x I - \partial_t \dot{D} \partial_t I $.
\end{flushleft}

\begin{flushleft}
\justify
Equations \eqref{A13} and \eqref{A12} give us
\begin{equation*} \label{A14}
\rho_0 \dfrac{\partial^2 \pi}{\partial t^2} - \dfrac{\partial }{\partial x} \left( c \dfrac{\partial \pi}{\partial x} \right) = \rho_0 \left( \dfrac{\partial F_0}{\partial t	} - \dfrac{\partial Q}{\partial x} \right)
\end{equation*}
\end{flushleft}

\begin{flushleft}
\justify
One possible Lagrangian $L^2$ to derive \eqref{A12} from is 
\begin{equation*} \label{A15}
L^2 = \dfrac{\rho_0}{2} \left( \dfrac{\partial \pi}{\partial t} \right)^2 - \dfrac{c}{2} \left( \dfrac{\partial \pi}{\partial x} \right)^2 - \dot{C} \pi
\end{equation*} 
where $\dot{C} = - \rho_0(\partial_t F_0 - \partial_x Q) $. 
\end{flushleft}

\begin{flushleft}
\justify
The Lagrangian $L^2$ is the basis for formulating the rest of the variational framework in the second iteration. While, as we see, the detailed expressions become increasingly more complex, they can be organized in an expected pattern as we move up from one iteration to the other.
\end{flushleft}

\clearpage

\nocite{*}
\bibliographystyle{vancouver}
\bibliography{Said_Arxiv_accepted}

\end{document}